\documentclass[sn-mathphys-num]{sn-jnl}

\usepackage{listings}
\usepackage{graphicx}
\usepackage{verbatim}
\usepackage{url}
\usepackage{siunitx}
\usepackage{xspace}
\usepackage{subcaption}
\usepackage{graphicx}
\usepackage{amsmath}
\usepackage{amssymb}
\usepackage[section]{placeins}
\usepackage{todonotes}
\usepackage{makecell}
\usetikzlibrary{fit,shapes.arrows}
\usepackage{float} 

\usepackage{multirow}%
\usepackage{mathrsfs}%
\usepackage[title]{appendix}%
\usepackage{xcolor}%
\usepackage{textcomp}%
\usepackage{manyfoot}%
\usepackage{algorithm}%
\usepackage{algorithmicx}%
\xdefinecolor{tured}{RGB}{153, 0, 0}
\xdefinecolor{goodgreen}{RGB}{0, 100, 0}
\xdefinecolor{tugreen}{RGB}{132, 184, 25}

\usepackage{amsthm}
\usepackage{amsmath}
\usepackage{amsfonts}
\usepackage{algorithm}
\usepackage{algpseudocode}
\newtheorem{theorem}{Theorem}
\newtheorem{definition}{Definition}
\newtheorem{lemma}{Lemma}

\usepackage[frozencache,cachedir=minted]{minted}
\definecolor{bgmint}{gray}{0.9}
\usepackage{booktabs}

\begin{document}
\parindent 0mm

\title[MCBench: A Benchmark Suite for Monte Carlo Sampling Algorithms]{MCBench: A Benchmark Suite for Monte Carlo Sampling Algorithms}

\author[1,3]{\fnm{Zeyu} \sur{Ding}}\email{zeyu.ding@tu-dortmund.de}
\equalcont{These authors contributed equally to this work.}
\author[2]{\fnm{Cornelius} \sur{Grunwald}}\email{cornelius.grunwald@tu-dortmund.de}
\equalcont{These authors contributed equally to this work.}
\author[1,3,4]{\fnm{Katja} \sur{Ickstadt}}\email{ickstadt@statistik.tu-dortmund.de}
\equalcont{These authors contributed equally to this work.}
\author[2,4]{\fnm{Kevin} \sur{Kröninger}}\email{kevin.kroeninger@cern.ch}
\equalcont{These authors contributed equally to this work.}
\author[2]{\fnm{Salvatore} \sur{La Cagnina}}\email{salvatore.lacagnina@tu-dortmund.de}
\equalcont{These authors contributed equally to this work.}

\affil[1]{\orgdiv{Department of Statistics}, \orgname{TU Dortmund University}, \orgaddress{\street{Vogelpothsweg 87}, \city{Dortmund}, \postcode{44227}, \country{Germany}}}
\affil[2]{\orgdiv{Department of Physics}, \orgname{TU Dortmund University}, \orgaddress{\street{Otto-Hahn-Straße 4}, \city{Dortmund}, \postcode{44227}, \country{Germany}}}
\affil[3]{
\orgname{Lamarr-Institute for Machine Learning and Artificial Intelligence},
\affil[4]{\orgdiv{TU Dortmund - Center for
Data Science and Simulation}, \orgname{TU Dortmund University}, \orgaddress{\street{August-Schmidt-Straße 4}, \city{Dortmund}, \postcode{44227}, \country{Germany}}}


\abstract{
\noindent \parindent 0mm
In this paper, we present MCBench, a benchmark suite designed to assess the quality of Monte Carlo (MC) samples. 
The benchmark suite enables quantitative comparisons of samples by applying different metrics, including basic statistical metrics as well as more complex measures, in particular the sliced Wasserstein distance and the maximum mean discrepancy. 
We apply these metrics to point clouds of both independent and identically distributed (IID) samples and correlated samples generated by MC techniques, such as Markov Chain Monte Carlo or Nested Sampling. Through repeated comparisons, we evaluate test statistics of the metrics, allowing to evaluate the quality of the MC sampling algorithms.

Our benchmark suite offers a variety of target functions with different complexities and dimensionalities, providing a versatile platform for testing the capabilities of sampling algorithms.
Implemented as a Julia package, MCBench enables users to easily select test cases and metrics from the provided collections, which can be extended as needed.
Users can run external sampling algorithms of their choice on these test functions and input the resulting samples to obtain detailed metrics that quantify the quality of their samples compared to the IID samples generated by our package. 
This approach yields clear, quantitative measures of sampling quality and allows for informed decisions about the effectiveness of different sampling methods.

By offering such a standardized method for evaluating MC sampling quality, our benchmark suite provides researchers and practitioners from many scientific fields, such as the natural sciences, engineering, or the social sciences with a valuable tool for developing, validating and refining sampling algorithms.}


\maketitle

\section{Introduction}
A large number of scientific fields use numerical methods to address mathematical problems which, in most cases, cannot be solved analytically. A common variant is the Monte Carlo (MC) method. It is based on drawing random numbers from  (probability) distributions defined by the problem at hand and the calculation of suitable quantities, e.g., expected values. The MC method is mostly used for integration, optimization and uncertainty propagation. An example for an intensive application of the MC method in the natural sciences is the simulation of scattering processes in particle physics, where -- experimentally -- two types of particles are repeatedly collided head-on under the same conditions and the resulting interactions are measured by large detectors. The processes and their resulting experimental signatures can be accurately modeled, but only be calculated numerically. Another example stems from the engineering sciences, particularly from civil engineering, where the tilting of the foundation in consolidation problems of soil is the quantity of interest. There, the tilting is calculated exactly by a finite element method for specific, fixed values of its inputs, the soil parameters. However, the soil parameters are subject to uncertainties, which leads to an uncertain result. To model this uncertainty, MC sampling in combination with a Kriging model is employed in frequentist as well as Bayesian analyses, see e.g. Refs.~\cite{williams2006gaussian,vanmeegen2025benchmarking}.

A key challenge when using the MC method is to obtain a sampling distribution that represents the target function sufficiently well. While drawing random numbers for some specific distributions is trivial, a variety of methods have been developed since the late 1940s to obtain samples from arbitrary distributions. Besides the brute force approach of the good old “hit \& miss” algorithm, some of the most notable examples of refined methods include importance sampling and other variance-reduction techniques, as well as variants of random walk-inspired Markov Chain Monte Carlo (MCMC), e.g. the traditional Metropolis algorithm~\cite{10.1063/1.1699114}. For most applications, it is desirable that the resulting sample of random numbers is independently and identically distributed (IID), so that all random numbers are mutually independent. Samples generated by MCMC algorithms, for example, are not IID and often show a significant amount of autocorrelation.  

In specific applications, the choice of the sampling algorithm depends on the characteristics of the probability distribution under investigation, in particular its dimension and general topology. Conveniently, a variety of software packages and tools have been made available in recent years that offer sampling algorithms either in a specific context, e.g. MCMC methods in conjunction with statistical inference such as \texttt{Stan}~\cite{carpenter2017stan} or \texttt{BAT.jl}~\cite{Schulz:2021BAT}, or as stand-alone numerical methods. Regardless of whether one uses existing tools or develops one's own software code, it is important to evaluate the quality of the generated samples. In particular, with high-dimensional distributions it is often not easy to assess whether the generated samples represent the target function in all corners of the phase space.

Although the need to compare samples from high-dimensional distributions has been recognized in various fields, e.g. in applications in particle physics~\cite{Grossi:2024axb}, there is as yet no general, domain-neutral tool that is modular in terms of test metrics, flexible in the definition of test cases and easy to use. This paper attempts to close this gap by introducing the MCBench package as a general benchmark suite developed specifically to evaluate the quality of MC sampling algorithms. The suite provides a selection of concrete test functions -- with strong variation in dimension and complexity -- from which the user can generate samples. The resulting point clouds are read into the software package and then compared with a large set of IID-produced samples of the same distributions. For this purpose, several uni- and multidimensional metrics with sensitivities to a wide range of different features are evaluated. The paper is structured as follows: after an outline of the basic design idea of the benchmark suite (Section~\ref{sec:design_ideas_and_work_flow}), the different test functions (Section~\ref{sec:target_functions}) and the various implemented metrics (Section~\ref{sec:metrics}) are described, followed by the implementation details (Section~\ref{sec:implementation_details}). Finally, selected problems are discussed in detail to demonstrate the application of the benchmark suite (Section~\ref{sec:walkthrough_example}).


\section{Design Ideas and Workflow} \label{sec:design_ideas_and_work_flow}



The benchmark suite evaluates the quality of a sampling algorithm by comparing sample distributions drawn from well-defined test functions using (a) the algorithm under study and (b) algorithms that produce IID-distributions. Several metrics are used for the comparison. The general workflow of the benchmark suite is illustrated in Figure~\ref{fig:flow_chart} and comprises three steps:

\paragraph{Choice of test function and generation of samples}
Users select a specific test function from the available list and implement it as a target function in a sampling code of their choice that provides an implementation of the sampling algorithm under test. The test functions vary in  dimensionality and complexity. A detailed list of the available test functions is provided in Section~\ref{sec:target_functions}. Example implementations of all test functions included in the list to be used with \texttt{Stan}, \texttt{PyMC} and \texttt{BAT.jl} are available at the Github repository of the project\footnote{\url{https://github.com/tudo-physik-e4/MCBench}}. Users generate samples from the target function which are then saved into a file.

\paragraph{Comparison of distributions}
The user-generated samples are imported into the benchmark suite and compared against IID samples from the same target function, which are automatically generated by the software. The comparison is done using a set of metrics detailed in Section~\ref{sec:metrics}. For the basic statistical measures described in Section~\ref{sec:basic_statistical_measures},the metrics are evaluated on $m$ different batches of $n$ IID samples each. The results are aggregated into histograms to obtain distributions of the expected values, with the distribution widths reflecting the statistical variation due to the limited sample size in each batch. The user-generated samples are also partitioned into $m$ batches, each with an effective sample size of $n$. The computation of the effective sample size is based on the sampling method and can be implemented individually. 
As a default, the ratio of the squared sum of the weights divided by the sum of their squares is used as an estimate.
The metrics are again evaluated for each of these batches, enabling a direct comparison between the distributions derived from the user-generated samples and those obtained from the IID samples. For the advanced statistical measures described in Section~\ref{sec:advanced_statistical_measures} that involve two-sample comparisons, we evaluate the metrics by performing IID vs. IID comparisons to generate the expected distribution, as well as IID vs. user-generated sample comparisons to assess the performance of the sampling algorithm under investigation.

\paragraph{Visualization of the results}
In the final comparison plots, we present the distributions of the metrics in a normalized way, see e.g. Fig.~\ref{fig:flow_chart} (bottom right). The mean value of each metric calculated from the set of IID-sampled distributions are centered around 0 (black vertical line). The regions corresponding to the $1~\sigma$-, $2~\sigma$-, and $3~\sigma$-ranges of each metric calculated from the IID-sampled distribution are colored green, yellow, and red, respectively. By overlaying the mean value and the standard deviation of the metric evaluated on the user-generated samples (marker with horizontal error bars), we provide a straightforward visual comparison between the IID and the user-generated samples.

\begin{figure}[h]
    \centering
    \includegraphics[width=\textwidth]{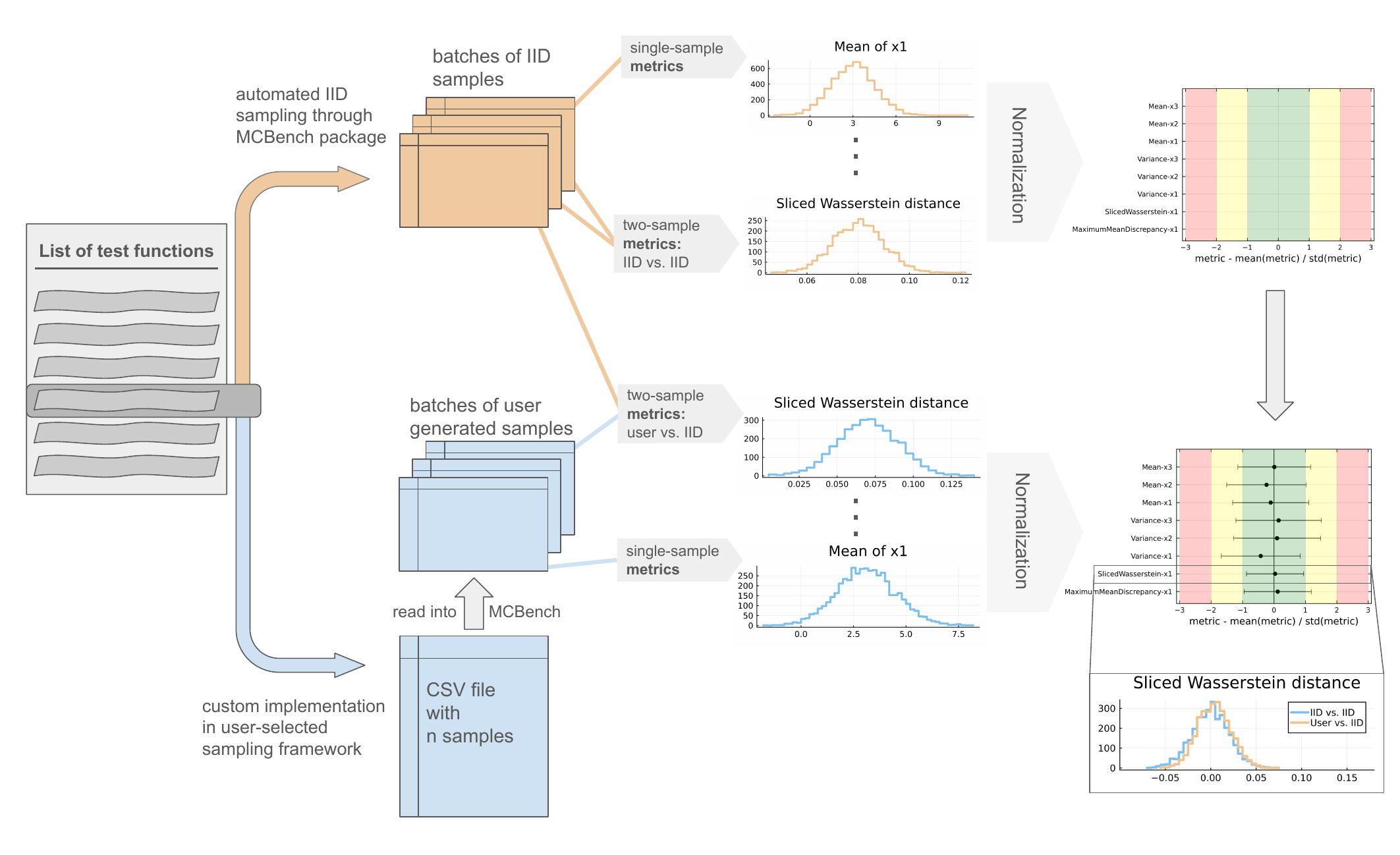}
    \caption{Illustration of the MCBench benchmark suite workflow.}
    \label{fig:flow_chart}
\end{figure}

\section{Target Functions} \label{sec:target_functions}

The benchmark suite provides a number of target functions that cover a wide range of complexities, including varying dimensionality, multi-modality, and correlations. They are chosen to represent practical applications or show specific features that challenge sampling algorithms. All target functions provided are IID-sampleable, which is a requirement for the benchmarking process. The list of target functions used is shown in Table~\ref{tab:test_cases}. We present two detailed walkthrough examples in Section~\ref{sec:walkthrough_example}, namely the one-dimensional Gaussian case and a three-dimensional bimodal distribution.

The target functions include standard normal distributions for multiple dimensionalities, correlated normal distributions, Cauchy distributions and mixture models which are used to create multimodal distributions.

In addition to the standard functions, we also include the eight schools example from the \texttt{posteriorDB} database Ref.~\cite{magnusson2024posteriordb}, which is a model used to estimate the effects of coaching programs on student performance.
The model consists of three hierarchical prior distributions:
\begin{align}
    \mu &= \mathcal{N}(0, 5) \nonumber \\
    \tau &= \text{Cauchy}(0, 5) \quad \text{with} \quad \tau \in (0, \infty)  \nonumber \\
    \theta_i &= \mathcal{N}(\mu, \tau)  \quad \text{with} \quad i = 1, \ldots, 8 \nonumber\,.
\end{align}
Given $\sigma_i$ and $\theta_i$, the likelihood is given by:
\begin{align}
    L(\theta, \sigma, y) = \prod_{i=1}^8 \mathcal{N}(y_i | \theta_i, \sigma_i) \, .
\end{align}
In order to compare the performance of the samplers, IID samples are drawn using a simple accept-reject algorithm for the eight schools example.
These test functions provide a basis for evaluating the performance of the samplers across a range of different distributions and complexities.
In addition to the implementation of the test functions in Julia, we also provide examples of how to use the test functions in PyMC and Stan.
These examples as well as supplementary information about the test cases, such as evaluation of the function for specific test points, are provided in the github repository for the benchmark suite~\footnote{\url{https://github.com/tudo-physik-e4/MCBench}}.

\begin{table}[h!]
\centering
\begin{tabular}{ccccc}
    \toprule
    \textbf{Name} & \textbf{Equation}  & \textbf{Properties} \\
    \midrule
    \makecell{Standard Normal 1D} & 
    $\begin{aligned}
    f(x \mid \mu, \sigma) &= \frac{1}{\sqrt{2 \pi} \sigma} \cdot e^{-\frac{(x-\mu)^2}{2\sigma^2}}
    \end{aligned}$ & 
    \makecell{Unimodal}\\
    
    \addlinespace
    
    \makecell{Standard Normal $k$D \\ uncorrelated} & 
    $\begin{aligned}
    \mathcal{N}_k(&\mu, \Sigma) = (2\pi)^{-k/2} \det(\Sigma)^{-1/2} \\
    &\cdot \exp\left( -\frac{1}{2} (x-\mu)^\top \Sigma^{-1} (x-\mu) \right)
    \end{aligned}$ &  
    \makecell{Unimodal, \\ $k=2,3,10,100$} \\
    
    





    \addlinespace

    \makecell{Standard Normal $k$D\\ weakly/strongly correlated} & 
    $\begin{aligned}
    \mathcal{N}_k(&\mu, \Sigma)\\
    \Sigma = r \cdot \mathbf{J} &+ (1-r) \cdot \mathbf{I}
    \end{aligned}$ & 
    \makecell{$r = 0.2, 0.9$ \\ $k=2,10,100$} \\

    \addlinespace



    \makecell{Mixture Normal kD\\ strongly correlated} & 
    $\begin{aligned}
    f(x \mid \mu, \Sigma) = 0.25 \; &\mathcal{N}_{k}(\mu, \Sigma) + 0.75 \; \mathcal{N}_{k}(-\mu, \Sigma)\\
    \Sigma = r &\cdot \mathbf{J} + (1-r) \cdot \mathbf{I}
    \end{aligned}$ &
    \makecell{Multimodal \\ $r = 0.9$,$k=3,10$} \\

    \addlinespace
 
    \makecell{Cauchy 1D} & 
    $\begin{aligned}
    f(x \mid x_0, \gamma) &= \frac{1}{\pi \gamma} \cdot \frac{1}{1 + \left( \frac{x-x_0}{\gamma} \right)^2}
    \end{aligned}$ & 
    \makecell{Unimodal} \\

    \addlinespace

    \makecell{Eight schools example} & 
    $\begin{aligned}
        L(\theta \mid y) = \prod_{i=1}^8 \frac{1}{\sqrt{2\pi \sigma_i^2}} \exp\left( -\frac{(y_i - \theta_i)^2}{2\sigma_i^2} \right)
    \end{aligned}$ & 
    \makecell{For more details \\ see Section~\ref{sec:target_functions}} \\
    
    \bottomrule
    \end{tabular}
\caption{List of test cases used in the benchmark suite including names, equations, and properties. The $\mathbf{J}$ refers to a matrix filled with ones while $\mathbf{I}$ refers to the identity matrix. More details about the test cases can be found in the documentation of the benchmark suite.}
\label{tab:test_cases}
\end{table}

\section{Metrics}\label{sec:metrics}

Our benchmark suite offers a collection of statistical distance measures, or metrics, to assess discrepancies between two sets of samples and to evaluate if they are drawn from the same target distribution. These metrics vary in their sensitivity to specific features of the distributions and their robustness in detecting discrepancies. We employ both basic, well-known metrics and more advanced statistical distance measures, which we briefly summarize in the following section.

\subsection{Basic Statistical Distance Measures}
\label{sec:basic_statistical_measures}
Traditional descriptive statistics, such as the mean, variance, and chi-square test statistics provide fundamental insights into the characteristics of a data set. 
These basic measures allow users to easily grasp the distribution's key features, including its central tendency (mean) and variability (variance). They are particularly well-suited for analyzing low-dimensional data or projections of high-dimensional data, offering a preliminary and intuitive understanding of the data distribution.
However, these basic statistical measures reach their limitations when analyzing high-dimensional or multimodal data. The overall structure and distributional characteristics of such data are often inadequately captured due to the so-called \emph{curse of dimensionality}. Therefore, introducing more complex statistical distance measures becomes necessary to examine high-dimensional data at a finer level of detail.

\subsection{Advanced Statistical Distance Measures}\label{sec:advanced_statistical_measures}
Complex distance metrics are specifically designed to measure the similarity of two distributions in high-dimensional and complex cases. In particular, such metrics not only refer to the overall shape of data distributions but also to particular differences between them, and hence, provide a more expressive method of comparison.  In general, statistical distance measures quantify the divergence between two probability distributions and are widely used in many data analysis and machine learning tasks. So far, we have implemented two advanced statistical distance measures: the (sliced) Wasserstein distance (SWD) and the maximum mean discrepancy (MMD). Further statistical distance metrics are planned to be implemented in the future.

\subsubsection{The Sliced Wassertein Distance}
The Wasserstein distance, or Earth Mover's distance, is a statistical measure of dissimilarity between two probability distributions. In the context of optimal transport theory, the Wasserstein distance can be understood as the minimum cost required to transport a mass from one probability distribution to another. This theory was first proposed by the French mathematician Gaspard Monge in 1781 \cite{monge1781mémoire} and later further popularized and developed by the Russian mathematician Leonid Kantorovich in the 20th century \cite{kantorovich1942translocation}. Recently, the Wasserstein distance has been widely used in different fields of machine learning and statistics. Examples include  Wasserstein-Generative-Adversarial-Networks~\cite{arjovsky2017wasserstein} and the approximation of the posterior distributions in Bayesian statistics \cite{ding2024scalable}.

According to Ref.~\cite{panaretos2019statistical}, the
$p$-Wasserstein distance between two probability measures $\mu$ and $\nu$ is given by 
\begin{equation*}
W_p(\mu, \nu) = \inf_{\substack{X \sim \mu \\ Y \sim \nu}} \left( \mathbb{E} \|X - Y\|^p \right)^{1/p} = \left( \inf_{\gamma \in \Gamma(\mu, \nu)} \int_{\mathcal{X} \times \mathcal{X}} \|x - y\|^p \, d\gamma(x, y) \right)^{1/p}, \quad p \ge 1.
\end{equation*}

Here, the element $\gamma \in \Gamma(\mu, \nu)$ can be viewed as a coupling of $\mu$ and $\nu$, which is a joint distribution with the given marginal distributions $\mu$ and $\nu$ on each axis. In the discrete case, this analytical expression can be interpreted as follows: for $\gamma \in \Gamma(\mu, \nu)$ and any pair of positions $(x, y)$, the value of $\gamma(x, y)$ represents the proportion of the mass of $\mu$ that is transferred from $x$ to $y$ in order to transform $\mu$ into $\nu$.

Although the Wasserstein distance can accurately measure the distance of two probability distributions, especially for its consideration of the geometric structure of the distributions, the computational complexity is generally large with $O(n^3\log n) $. Also, it suffers the curse of dimensionality, where the sample complexity increases with the dimension $d$ by $O(n^{-1/d})$, see~Ref.~\cite{nguyen2022amortized}. This means that when $d$ increases, an exponential increase in sample size is required to maintain the same estimation accuracy.

To overcome the problematic scaling for high-dimensional test cases, the benchmark suite uses the sliced Wasserstein distance. It approximates the difference between high-dimensional distributions by first projecting them onto one dimension along various random directions, then computing the one-dimensional Wasserstein distance for each projection, and finally averaging these one-dimensional distances. 
This method was first proposed in Ref.~\cite{rabin2012wasserstein} using a random projection approach.
The  random-projection-based SWD is defined as: \\

\begin{definition}[Sliced Wasserstein Distance]
    
Let $ \mu $ and $ \nu $ be probability distributions defined on $ \mathbb{R}^d $ with the sliced Wasserstein distance defined as:

\begin{equation}
\text{SWD}_p(\mu, \nu) = \left( \int_{\mathbb{S}^{d-1}} W_p^p\left( P_\theta \mu, P_\theta \nu \right) \, d\theta \right)^{1/p} \, .
\end{equation}
\end{definition}

Here, $ \mathbb{S}^{d-1} $ is the unit sphere in $ \mathbb{R}^d $.
$ P_\theta \mu $ and $ P_\theta \nu $ denote the one-dimensional projected distributions obtained by projecting $ \mu $ and $ \nu $ onto the direction $ \theta $ respectively.
The $W_p\left( P_\theta \mu, P_\theta \nu \right) $ is the $p$-Wasserstein distance projected onto the direction $ \theta $.

Since it is not possible (and not necessary) to project onto all possible $ \theta $ directions, the Monte Carlo method is used to approximate the above integrals with a finite number of randomly sampled directions $ \{ \theta_i \}_{i=1}^L $:

\begin{equation}
\text{SWD}_p(\mu, \nu) \approx \left( \frac{1}{L} \sum_{i=1}^L W_p^p\left( P_{\theta_i} \mu, P_{\theta_i} \nu \right) \right)^{1/p}.
\end{equation}

For a detailed description of the algorithm for sliced Wasserstein distance, we refer to Algorithm~\ref{algo:swdist} in Appendix~\ref{Appendix:Algo}, for more technical details of the SWD we refer to Appendix Section~\ref{subsec:WS}.

\subsubsection{Maximum Mean Discrepancy}

The maximum mean discrepancy (MMD) measures the distance between two distributions by comparing moments in a high-dimensional feature space. It maps the data into a Reproducing Kernel Hilbert Space (RKHS) and computes the distance between the mean embeddings of the two distributions to measure their divergence. The advantage of MMD is its effectiveness in handling high-dimensional data and its good performance with complex distributions. Ref.~\cite{gretton2012kernel} defines it as \\

\begin{definition}[Maximum Mean Discrepancy in RKHS]
Let $(\mathcal{H}, k)$ be a reproducing kernel Hilbert space with kernel function $k: \mathcal{X} \times \mathcal{X} \to \mathbb{R}$. For two probability measures $p$ and $q$ defined on a metric space $\mathcal{X}$, assuming $\mathbb{E}_x[k(x,x)] < \infty$ and $\mathbb{E}_y[k(y,y)] < \infty$, the maximum mean discrepancy is defined as:

\begin{equation}
\text{MMD}[p,q] = \|\mu_p - \mu_q\|_{\mathcal{H}}
\end{equation}

where $\mu_p$ and $\mu_q$ are the mean embeddings of distributions $p$ and $q$ respectively in $\mathcal{H}$, defined as:

\begin{equation}
\mu_p = \mathbb{E}_{x\sim p}[k(x,\cdot)] \quad \text{and} \quad \mu_q = \mathbb{E}_{y\sim q}[k(y,\cdot)].
\end{equation}

The squared MMD can be expanded in terms of kernel functions as:

\begin{equation}
\begin{split}
\text{MMD}^2[p,q] = \mathbb{E}_{x,x'\sim p}[k(x,x')] + \mathbb{E}_{y,y'\sim q}[k(y,y')] \\
- 2\mathbb{E}_{x\sim p,y\sim q}[k(x,y)].
\end{split}
\end{equation}

\end{definition}

MMD has been widely applied in numerous machine learning tasks, including two-sample tests, domain adaptation, and generative model evaluation. Compared to the Wasserstein distance, MMD has an advantage in high-dimensional data processing because it can flexibly capture the higher-order structural features of the data through the choice of kernel functions. We refer to Section~\ref{subsec:MMD} of the appendix for more technical details of the MMD.

The choice of kernel function in MMD significantly affects its performance and ability to characterize differences between distributions. Commonly used kernels include the Gaussian, Laplacian, Polynomial, Linear, and Sigmoid kernels. The benchmark suite uses a Gaussian kernel $k_G(x, y)$ that is defined as:

\begin{equation}
k_G(x, y) = \exp\left(-\frac{\|x-y\|^2}{2\sigma^2}\right)  \, . 
\end{equation}

The parameter $\sigma$ (the kernel width) plays a critical role in determining its sensitivity. A commonly used approach for selecting $\sigma$ is the \textit{median heuristic}, which sets $\sigma$ to the median of pairwise distances among data points. 
%
For a more detailed discussion on the characteristic kernels and the parameter choices of the MMD, we refer to Refs.~\cite{sriperumbudur2010hilbert,gretton2012optimal}. \\

In practice, the empirical squared MMD can be computed as follows:

\begin{equation}
\begin{split}
    \text{MMD}^2(X, Y) = & \frac{1}{n^2}\sum_{i,j=1}^n k(x_i, x_j) + \frac{1}{m^2}\sum_{i,j=1}^m k(y_i, y_j) \\ & - \frac{2}{nm}\sum_{i=1}^n\sum_{j=1}^m k(x_i, y_j),
\end{split}
\end{equation}
where $k(\cdot,\cdot)$ represents the kernel function, $X = {x_1, x_2, \dots, x_n}$ contains $n$ data points, and $Y = {y_1, y_2, \dots, y_m}$ contains $m$ data points.

Despite its theoretical elegance and ability to detect differences between distributions, the practical application of MMD faces significant computational challenges, especially when dealing with large-scale datasets.

Therefore, in practice, a common approach is to approximate the kernel calculations. A popular method is to use the Random Fourier Features (RFF). The theoretical justification for RFF comes from Bochner's theorem \cite{bochner1959lectures}. For a detailed theoretical explanation of Bochner's theorem and RFF, we refer to Appendix \ref{subsec:MMD}.

Both, the original MMD calculation for the Gaussian kernel as well as its corresponding RFF approximation are implemented in the benchmark suite. It can be shown that the computational complexity for the original Gaussian MMD is $O(d \cdot mn)$. If $n \approx m$, the computational complexity becomes $O(d \cdot n^2)$. This is because computing all pairwise kernel evaluations requires $O(n^2)$ operations, and each kernel computation involves $O(d)$ time due to the dimensionality of the data.

For large sample sizes, the memory and computational requirements would be substantial. Therefore, we recommend using the RFF to obtain the approximated MMD for large sample sizes, with a computational complexity of $O(D(n + m)d)$, where $D$ is the number of Random Fourier Features. In practice, if $D \ll n$, the computational time and memory requirements would be significantly reduced. For the detailed implementation of the two algorithms, we refer to Algorithm~\ref{algo:MMD} and Algorithm~\ref{algo:RFFMMD} in Appendix~\ref{Appendix:Algo}.
 
\section{Implementation Details}\label{sec:implementation_details}
We have developed our benchmark suite in the Julia programming language, which combines high performance with user-friendly syntax. Julia uses a just-in-time compiler that ensures efficient execution of computationally intensive tasks such as MC sampling. 
Its syntax is similar to that of the Python programming language, featuring a script-like structure and dynamic typing, making it easy to learn and straightforward to use for scientific applications.
Julia's rich ecosystem of packages for numerical calculations and statistical analysis makes it an ideal choice for implementing our benchmark suite.

The benchmark suite comes in the form of a Julia package and can be installed via the Julia package manager.
It offers three main objects, corresponding to the three essential components needed for benchmarking: \emph{test case}, \emph{sampler}, and \emph{metrics}. The implementation details of these components are discussed in the following.

The \emph{Test case} objects describe the selected IID-sampleable target distributions. 
These objects are designed to be compatible with the \texttt{Distributions.jl} package \cite{JSSv098i16,Distributions.jl-2019}, which makes them easy to use as parts of \emph{test cases}. 
Typically, the user will choose a test case from the list described in Section \ref{sec:target_functions}. But also custom test cases can be implemented using the \texttt{TestCase} datatype and linking to distributions or mixture model distributions from the \texttt{Distributions.jl} together with the chosen boundaries of the free parameters. For more complex or custom target functions that are not part of the \texttt{Distributions.jl} package, users can also define custom \texttt{Target} datatypes. In such a case, the user needs to implement the \texttt{Base.rand} and \texttt{Distributions.logpdf} functions for these types, which are needed for generating IID samples. 
Most of the target functions presented in Section \ref{sec:target_functions} are implemented using the \texttt{Distributions.jl} package. The \emph{Eight Schools} problem is implemented with a custom target function based on \texttt{posteriorDB} derived using \texttt{Stan} and \texttt{BridgeStan.jl} \cite{Roualdes2023}.

The \emph{sampler} object consists of information about the selected sampling algorithm and its properties. 
Samplers can be added by defining new structs as subtypes of the \texttt{SamplingAlgorithm} datatype and providing a specific implementation of the \texttt{sample} functions. 
These types provide the functionality to define a \texttt{TestCase}, a sampler and a number of samples, and they return a set of samples as a \texttt{DensitySampleVector} object using the desired sampler.

As an alternative to sampling the distributions on the fly by directly interfacing to a sampling framework, our benchmark suite allows to read in pre-sampled data using the \texttt{FileBasedSampler} datatype and a user-defined path to either a directory or single file. 
The samples should be stored in a file format that can be read by the \texttt{FileBasedSampler} object, e.g. the CSV or HDF5 format.
Samples can also be provided directly as a \texttt{DensitySampleVector} object.
This allows the user to use custom formats supported by any other Julia package, such as \texttt{JLD2.jl}~\footnote{\hyperlink{https://github.com/JuliaIO/JLD2.jl}{https://github.com/JuliaIO/JLD2.jl}} , and \texttt{HDF5.jl}~\footnote{\hyperlink{https://github.com/JuliaIO/HDF5.jl}{https://github.com/JuliaIO/HDF5.jl}}.  

The \emph{metric} objects contain the result of the evaluated metric and its properties.
Each of the metrics introduced in Section \ref{sec:metrics} is its own type and extends the abstract type \texttt{TestMetric}.
The implementation of the metric requires a \texttt{calcmetric} function that takes a \texttt{DensitySampleVector} and returns a vector of metrics itself.
In the case of two sample tests, the metric type is \texttt{TwoSampleTestMetric} and the function receives two \texttt{DensitySampleVector} instances.
The evaluation of a metric on the sample batches is performed in parallel using the \texttt{Folds.jl}~\footnote{\hyperlink{https://github.com/JuliaFolds/Folds.jl}{https://github.com/JuliaFolds/Folds.jl}} package, which allows for efficient computations on multiple threads.
The test statistics for each metric and test case can be built iteratively and are written into a \texttt{JSON} file for further analysis.
The calculation of the Wasserstein Distance is based on the \texttt{R} implementation of the \texttt{transport} package~\cite{Rtransport}.
For the MMD, the \texttt{IPMeasures.jl}~\footnote{\hyperlink{https://github.com/JuliaFolds/IPMeasures.jl}{https://github.com/JuliaFolds/IPMeasures.jl}} package is used while the \texttt{Distances.jl}~\footnote{\hyperlink{https://github.com/JuliaFolds/Distances.jl}{https://github.com/JuliaFolds/Distances.jl}} package provides the calculation for the kernel width.

As an example, we can use the following code snippet to test the Metropolis-Hastings sampler of the \texttt{BAT.jl} package on a standard normal distribution in three dimensions which is provided as a test case. 
The \texttt{marginal\_mean}, \texttt{marginal\_variance} and \texttt{sliced\_wasserstein\_distance} metrics are used for the evaluation.
\begin{minted}[baselinestretch=1.2,bgcolor=bgmint,fontsize=\footnotesize]{julia}      
    testcase = Standard_Normal_3D_Uncorrelated
    metrics = [marginal_mean(),marginal_variance(),sliced_wasserstein_distance()]
    sampler = BATMH()
\end{minted}
Alternatively, the user can use the \texttt{FileBasedSampler} to read in pre-sampled data:
\begin{minted}[baselinestretch=1.2,bgcolor=bgmint,fontsize=\footnotesize]{julia}      
    sampler = FileBasedSampler("samples.csv")
\end{minted}
After defining the test case, the set of metrics used, and the sampler, the user can build the test statistic and plot the metrics using the following commands:
\begin{minted}[baselinestretch=1.2,bgcolor=bgmint,fontsize=\footnotesize]{julia}      
    build_teststatistic(testcase, metrics, n=100, n_samples=10^5)
    build_teststatistic(testcase, metrics, n=100, n_samples=10^5, s=sampler)
    plot_metrics(testcase,metrics,BATMH())
\end{minted}

The resulting plot will show the distribution of the metrics for the given test case and sampler compared to IID samples from the target distribution as will be demonstrated in the following section. \\

\section{Walkthrough Example}\label{sec:walkthrough_example}

We demonstrate the usage of the benchmark suite by evaluating the performance of the default Metropolis-Hastings sampler provided by the \texttt{BAT.jl} package.
The evaluation is carried out on the basis of two test cases of varying complexity: A basic example using a standard normal distribution in three dimensions and a more complex, multimodal mixture model of two correlated normal distributions in three dimensions.

\subsection{Basic Example: Standard Normal Distribution in 3D}
The test case is a three-dimensional normal distribution without correlation. It is expressed using the \texttt{Distributions.jl} package as follows, which allows the benchmark suite to generate IID samples on the fly:
\begin{minted}[baselinestretch=1.2,bgcolor=bgmint,fontsize=\footnotesize]{julia}      
    f = MvNormal(zeros(3), I(3))
    bounds = NamedTupleDist(x = [-10..10 for i in 1:3])
    Standard_Normal_3D_Uncorrelated = Testcases(f,bounds,3,"Normal-3D-Uncorrelated")
\end{minted}
Then, the metrics and the sampler to be used for this test case are defined:
\begin{minted}[baselinestretch=1.2,bgcolor=bgmint,fontsize=\footnotesize]{julia}      
    metrics = [marginal_mean(),marginal_variance(),
              sliced_wasserstein_distance(),maximum_mean_discrepancy()]
    sampler = BATMH(n_steps=10^5, nchains=10)
\end{minted}
In this case, the sampler is directly interfaced to the \texttt{BAT.jl} package.
Alternatively, when using a different sampling software package, it is possible to load the sample using the \texttt{FileBasedSampler} as follows :
\begin{minted}[baselinestretch=1.2,bgcolor=bgmint,fontsize=\footnotesize]{julia}      
    sampler = FileBasedSampler("samples.csv")
\end{minted}
After defining the test case, metrics, and sampler, the test statistic for the IID samples and the samples generated by the Metropolis-Hastings sampler are generated:
\begin{minted}[baselinestretch=1.2,bgcolor=bgmint,fontsize=\footnotesize]{julia}      
    teststatistics_IID = build_teststatistic(Standard_Normal_3D_Uncorrelated, metrics, 
                        n=100, n_steps=10^5, n_samples=10^5)
    teststatistics_MH = build_teststatistic(Standard_Normal_3D_Uncorrelated, metrics, 
                        n=100, n_steps=10^5, n_samples=10^5, s=sampler)
\end{minted}
As an example, a one-dimensional marginalized distribution for the IID-samples (blue) and the samples using using the Metropolis-Hastings algorithm (red) are shown in Figure~\ref{fig:3dnormal_samples_distributions}. No obvious difference between the two distributions is visible by eye.
For the metrics, 100 batches of $10^5$ samples each are used to evaluate the mean values and variances of all marginal distributions as well as the SWD and the MMD.
Finally, an overview of the metrics is generated comparing the samples from the Metropolis-Hastings sampler with the IID samples:
\begin{minted}[baselinestretch=1.2,bgcolor=bgmint,fontsize=\footnotesize]{julia}      
    plot_metrics(Standard_Normal_3D_Uncorrelated,metrics,sampler)
\end{minted}

\begin{figure}
    \centering
    \begin{subfigure}[t]{0.48\textwidth}
        \centering
        \includegraphics[width=\textwidth]{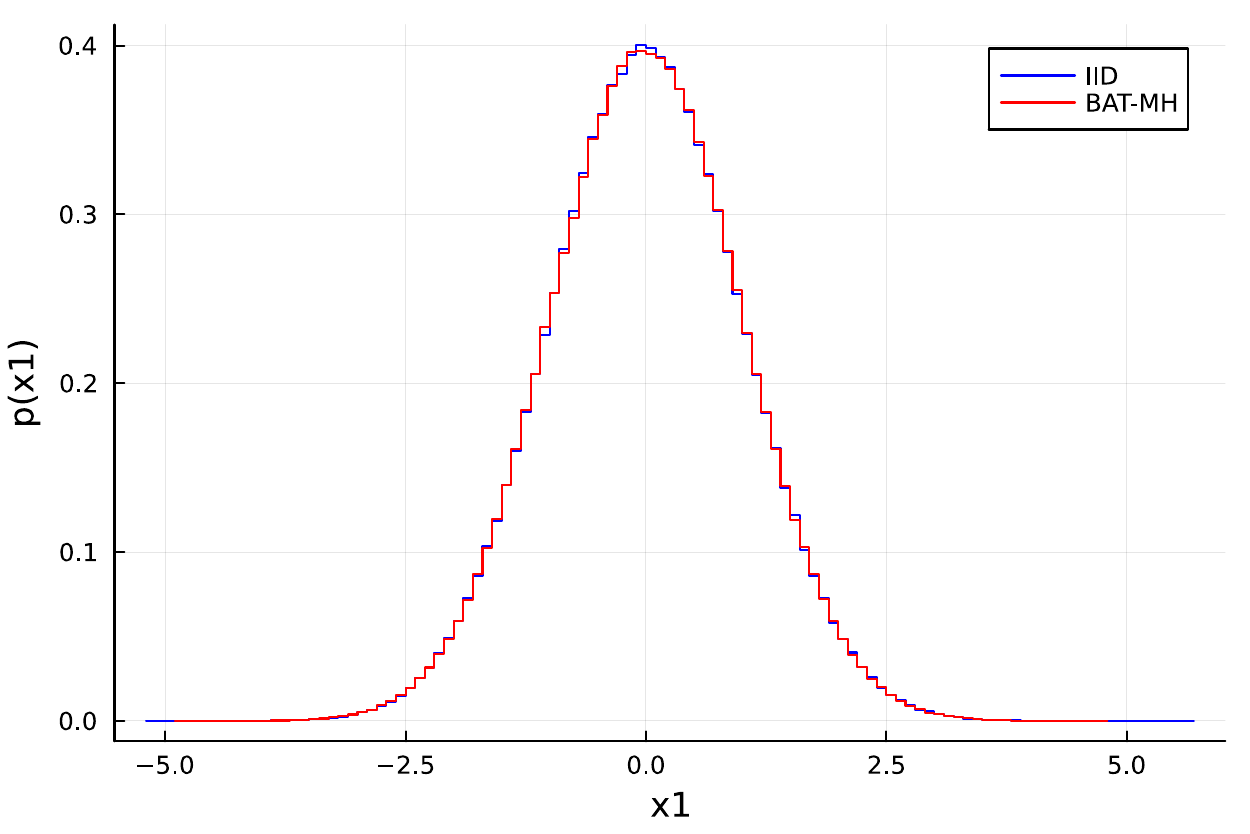}
        \caption{Uncorrelated 3D standard normal distribution.}
        \label{fig:3dnormal_samples_distributions}
    \end{subfigure}
    \hfill
    \begin{subfigure}[t]{0.48\textwidth}
        \centering
        \includegraphics[width=\textwidth]{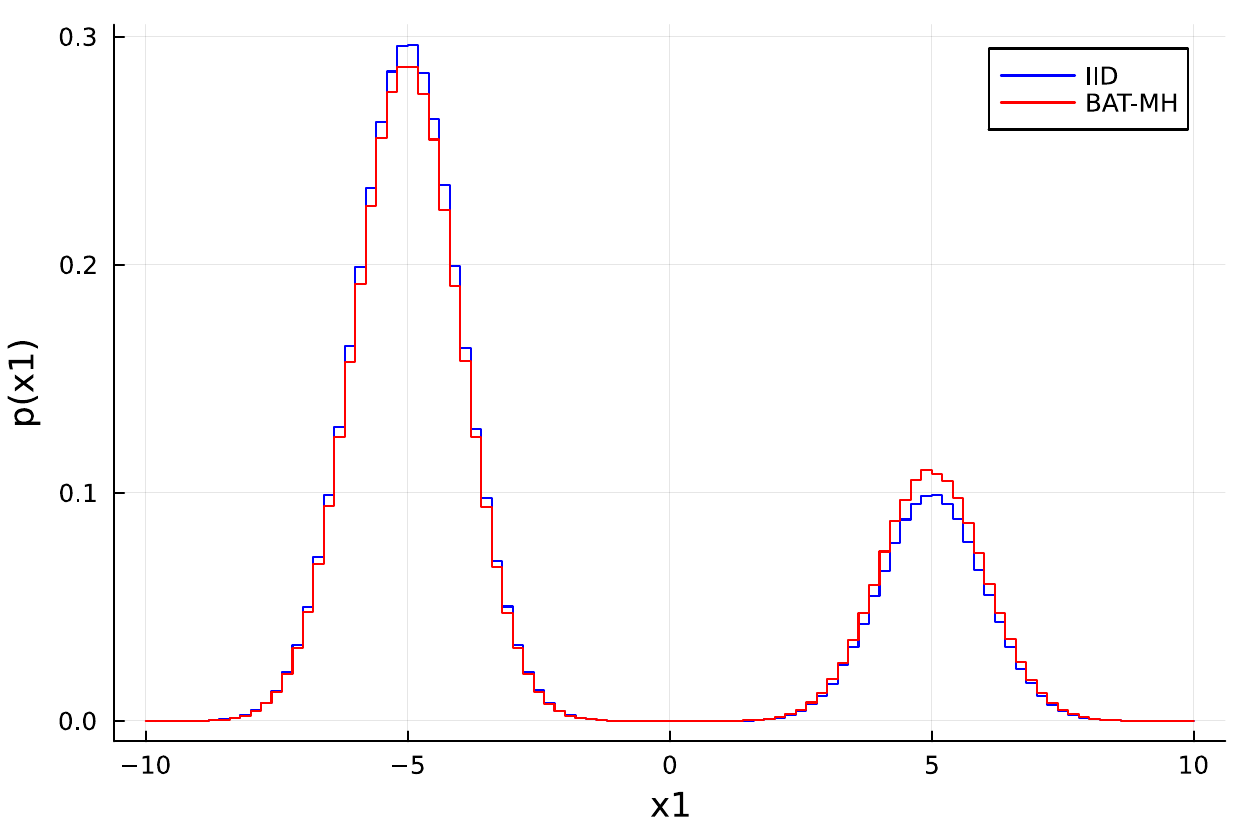}
        \caption{Mixture model with two 3D correlated normal distributions.}
        \label{fig:3dnormal_corr_samples_distributions}
    \end{subfigure}
    \caption{One-dimensional marginalized distributions for both test functions used in the examples using both IID sampling (blue) and Metropolis-Hastings (red).}
    \label{fig:samples_distributions}
\end{figure}

Figure~\ref{fig:metrics_3dnormal_example}  shows the expected mean values (black vertical line at 0) and standard deviations (colored regions for one, two, and three standard deviations) of each metric evaluated with the IID samples and compares them with the mean values (marker) and standard deviations (horizontal error bars) of the metrics evaluated with the samples provided by the Metropolis-Hasting sampler. In this concrete example, the metrics derived with the Metropolis-Hastings sampler roughly match the expectations in a sense that the mean values and standard deviations are close. A slight overestimation of the variance of the metrics is observed and is due to an overestimation of the effective sample size. The effective sample size is smaller than the actual sample size because of the autocorrelation of the samples that is intrinsically introduced by the Metropolis-Hastings algorithm. 

\begin{figure}
    \centering
    \includegraphics[width=0.8\textwidth]{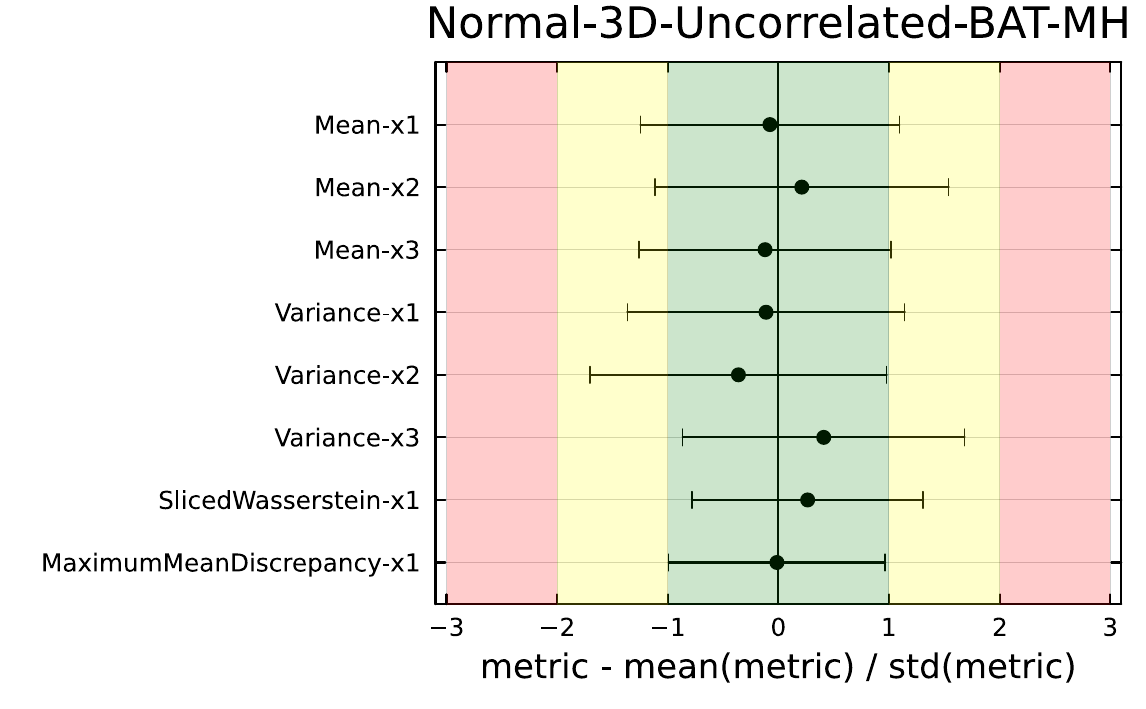}
    \caption{Mean values and standard deviations of several metrics calculated with IID samples (colored bands) and samples produced with the Metropolis-Hastings sampler implemented in BATjl. The samples are drawn from an uncorrelated standard normal distribution in three dimensions. The metrics are normalized to the mean and variance of the IID samples from the target distribution.}
    \label{fig:metrics_3dnormal_example}
\end{figure}

For a closer look on the  metrics, the distribution of the metrics for the samples generated by the Metropolis-Hastings sampler and the IID sampler can be plotted:
\begin{minted}[baselinestretch=1.2,bgcolor=bgmint,fontsize=\footnotesize]{julia} 
    plot_teststatistic(Standard_Normal_3D_Uncorrelated,metrics[1],sampler,nbins=20)
    plot_teststatistic(Standard_Normal_3D_Uncorrelated,metrics[3],sampler,nbins=20)
\end{minted}
The distributions for the mean of the first dimension (left) and the SWD (right) are shown in Figure~\ref{fig:metrics_3dnormal_example_detail} and provide a detailed view of the distribution of the metrics for the Metropolis-Hastings sampler (red) compared to the IID samples (blue). 
\begin{figure}
    \centering
    \begin{minipage}{0.48\textwidth}
        \centering
        \includegraphics[width=\textwidth]{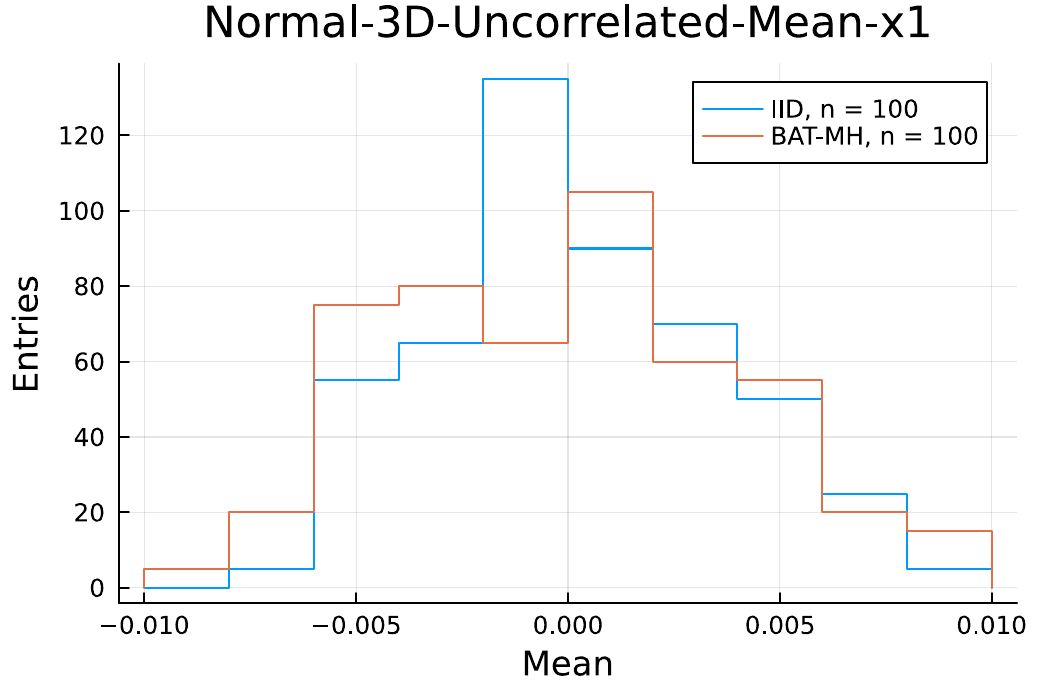}
        \caption*{(a) Mean of the marginalized distribution of the first dimension.}
    \end{minipage}
    \hfill 
    \begin{minipage}{0.48\textwidth}
        \centering
        \includegraphics[width=\textwidth]{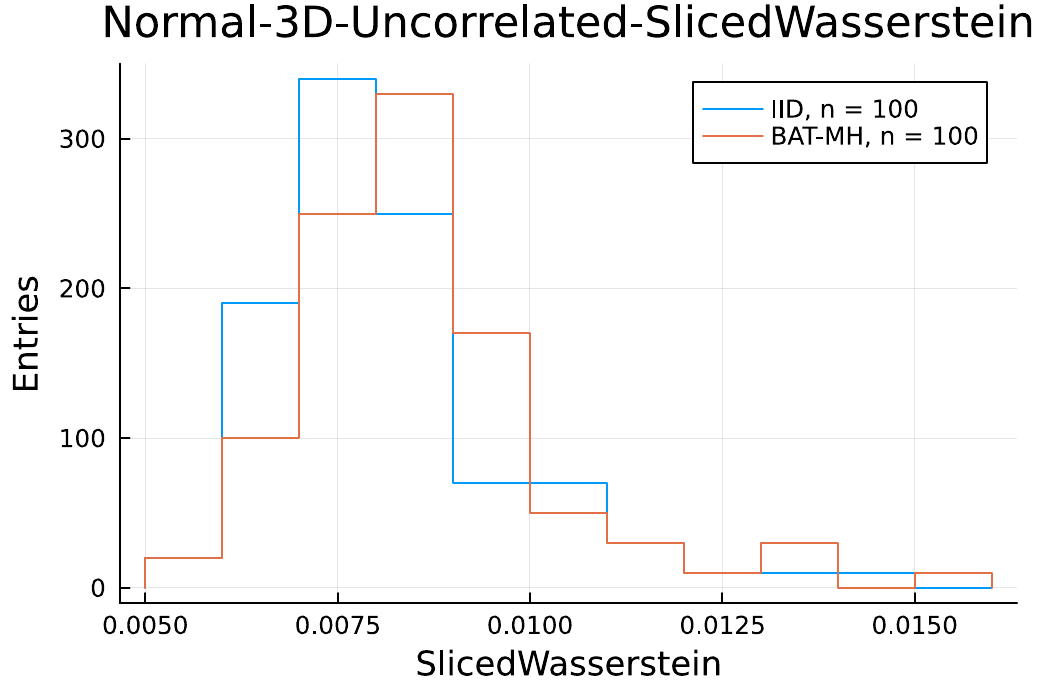}
        \caption*{(b) Sliced Wasserstein distance for the 3D distribution.}
    \end{minipage}
    \caption{Example of the distributions of the metrics for the Metropolis-Hastings sampler and the IID samples for a standard normal distribution in 3D.}
    \label{fig:metrics_3dnormal_example_detail}
\end{figure}

\FloatBarrier
\subsection{Complex Example: Mixture Model of Correlated Normal Distributions in 3D}
The second test case is a mixture model of two correlated normal distributions in 3D.
The definition of the test case can be expressed using the \texttt{Distributions.jl} package as follows:
\begin{minted}[baselinestretch=1.2,bgcolor=bgmint,fontsize=\footnotesize]{julia}      
    r = 5
    f1 = MvNormal(r*ones(10), ones(10,10)*0.9 + I(10)*0.1)
    f2 = MvNormal(-r*ones(10), ones(10,10)*0.9 + I(10)*0.1)
    f = MixtureModel([f1,f2], [0.25, 0.75])
    bounds = NamedTupleDist(x = [-100..100 for i in 1:10])
    normal_3d_multimodal_10std = Testcases(f,bounds,10,"Normal-3D-Multimodal-10std")
\end{minted}
The generation of samples, the evaluation of the metrics, and the generation of the test statistics are performed analogously to the basic example. As an example, a one-dimensional marginalized distribution of the samples for both methods is shown in Figure~\ref{fig:3dnormal_corr_samples_distributions}. 
While the samples might look similar at first glance the relative strength of the magnitude of the modes is significantly different. This effect is observed in all one-dimensional marginals.
The resulting distribution of the metrics for the example are shown in 
Figure~\ref{fig:metrics_3dnormal_corr_example}.
\begin{figure}[b]
    \centering
    \includegraphics[width=0.8\textwidth]{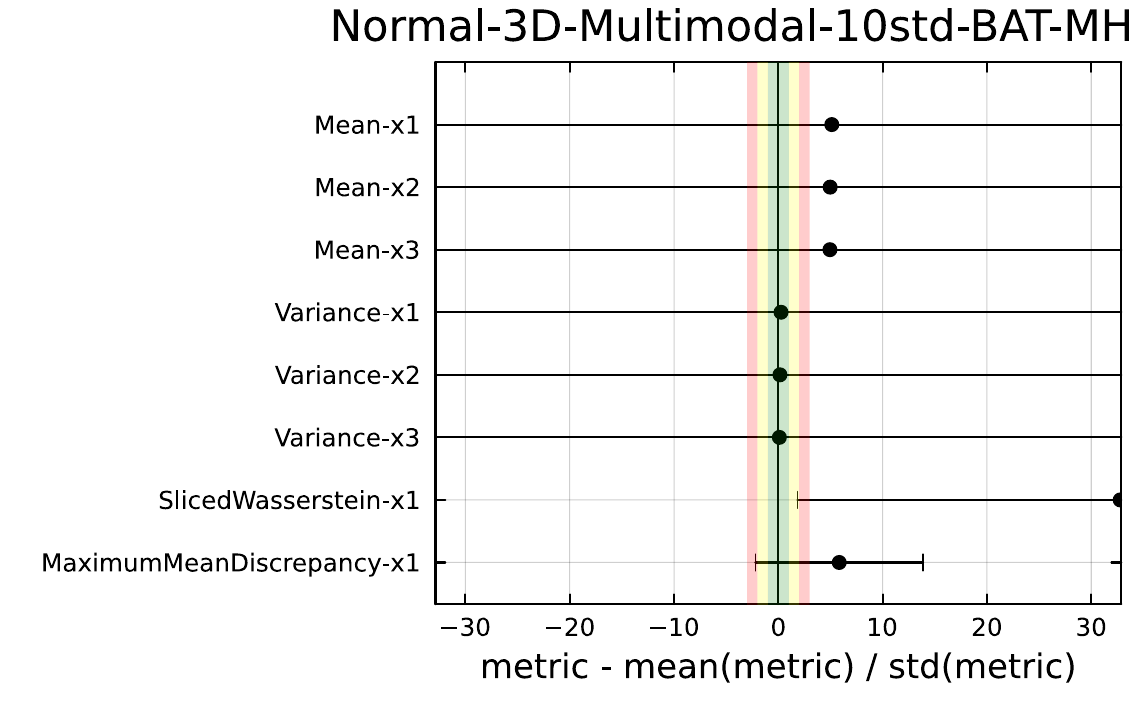}
    \caption{Example of metric distributions for the Metropolis-Hastings sampler on a mixture model of correlated normal distributions in 3D. The metrics are normalized to the mean and variance of the IID samples from the target distribution.}
    \label{fig:metrics_3dnormal_corr_example}
\end{figure}
The mean value and standard deviations of almost all metrics derived from samples obtained with the Metropolis-Hastings algorithm clearly deviate from those derived from the IID samples. This shows that the mixture model represents a challenge for the sampling algorithm compared to the basic example. This is to be expected, as the Metropolis-Hastings sampler is known to have difficulties sampling multimodal distributions, as it can get stuck in local modes. As a result, the sampler often does not meet the convergence criteria and therefore draws a set of samples that does not represent the target distribution. 
This example demonstrates the capability of the benchmark suite to evaluate the performance of samplers on a wide range of test cases with varying complexities.

\section{Conclusion}

In this paper, we have proposed a flexible and easy-to-use benchmark suite to evaluate the quality of MC sampling algorithms. The software package MCBench enables quantitative comparisons of samples drawn from well-defined multidimensional distributions by evaluating a variety of metrics. In addition to basic statistics, more sophisticated statistical measures such as the SWD and the MMD are used. The benchmark suite also features visualization and detailed information on the comparison of the samples.

In the future, we will increase the sensitivity of the benchmark suite by expanding the list of target functions and by increasing the number of metrics. In particular, we plan to include target functions with special features, such as the Rosenbrock function and the Himmelblau function, which serve as proxies for more complex scenarios and multimodal functions. In addition, we plan to explore more sophisticated distance metrics to further improve the accuracy and scope of sample quality assessment. In particular, we will explore other MMD kernels to test their effectiveness and applicability in different sampling algorithms. These additions will complete our benchmark suite and help to qualitatively improve the development and optimization of MC sampling algorithms.

\bmhead{Supplementary information}


\bmhead{Acknowledgments}
The authors are supported by the German Federal Ministry of Education and Research
(BMBF) in the ErUM-Data action plan via the KISS consortium (Verbundprojekt 05D2022).



\bibliography{references.bib}

\clearpage
\pagebreak

\begin{appendices}
\section{Algorithms} \label{Appendix:Algo}

\begin{algorithm}[H]
\caption{Calculating the sliced Wasserstein distance}
\label{algo:swdist}
\begin{algorithmic}[1]
\Require Two sample sets: $ X = \{ x_j \}_{j=1}^N \sim \mu $, $ Y = \{ y_j \}_{j=1}^N \sim \nu $; Number of projections $ L $; Order of distance $ p $
\Ensure Approximate sliced Wasserstein distance $ \text{SWD}_p(\mu, \nu) $
\State Initialize accumulation variable $ S \leftarrow 0 $
\For{ $ i = 1 $ to $ L $ }
    \State Randomly sample a unit vector from the unit sphere $ \mathbb{S}^{d-1} $ $ \theta_i $
    \State Calculate the projection:
    \[
    \alpha_j = \theta_i^\top x_j, \quad \beta_j = \theta_i^\top y_j, \quad \text{for } j = 1, 2, \dots, N
    \]
    \State Sort $ \{ \alpha_j \} $ and $ \{ \beta_j \} $ in ascending order, respectively, to obtain the sorted sequences $ \{ \alpha_{(j)} \} $ and $ \{ \beta_{(j)} \} $
    \State Compute the one-dimensional Wasserstein distance:
    \[
    W^{(i)} = \left( \frac{1}{N} \sum_{j=1}^N \left| \alpha_{(j)} - \beta_{(j)} \right|^p \right)^{1/p}
    \]
    \State Update accumulation variable $ S \leftarrow S + \left( W^{(i)} \right)^p $
\EndFor
\State Calculate SWD:
\[
\text{SWD}_p(\mu, \nu) = \left( \frac{S}{L} \right)^{1/p}
\]
\State \Return $ \text{SWD}_p(\mu, \nu) $
\end{algorithmic}
\end{algorithm}

\begin{algorithm}[H]
\caption{Calculating maximum mean discrepancy using Gaussian Kernel}
\label{algo:MMD}
\begin{algorithmic}[1]
\Require Two sample sets: $\{ x_i \}_{i=1}^n \sim p(x)$ and $\{ y_j \}_{j=1}^m \sim q(y)$; Gaussian kernel bandwidth parameter $\sigma$
\Ensure Empirical MMD value $\text{MMD}$; median heuristic calculation of $\sigma$ is also possible
\State \textbf{Step 1: Compute Kernel Matrices}
\State Compute the kernel matrix $K_{XX}$ where $(K_{XX})_{i,j} = \exp\left(-\frac{\|x_i - x_j\|^2}{2\sigma^2}\right)$ for all $i, j \in \{1, \dots, n\}$
\State Compute the kernel matrix $K_{YY}$ where $(K_{YY})_{i,j} = \exp\left(-\frac{\|y_i - y_j\|^2}{2\sigma^2}\right)$ for all $i, j \in \{1, \dots, m\}$
\State Compute the cross-kernel matrix $K_{XY}$ where $(K_{XY})_{i,j} = \exp\left(-\frac{\|x_i - y_j\|^2}{2\sigma^2}\right)$ for all $i \in \{1, \dots, n\}$ and $j \in \{1, \dots, m\}$

\State \textbf{Step 2: Calculate MMD Value}
\State Calculate the MMD squared:
\[
\text{MMD}^2 = \frac{1}{n^2} \sum_{i=1}^n \sum_{j=1}^n (K_{XX})_{i,j} + \frac{1}{m^2} \sum_{i=1}^m \sum_{j=1}^m (K_{YY})_{i,j} - \frac{2}{n m} \sum_{i=1}^n \sum_{j=1}^m (K_{XY})_{i,j}
\]

\State \textbf{Step 3: Return Result}
\State \Return $\text{MMD} = \sqrt{\text{MMD}^2}$

\end{algorithmic}
\end{algorithm}

\begin{algorithm}[H]
\caption{Calculating maximum mean discrepancy using random Fourier features}
\label{algo:RFFMMD}
\begin{algorithmic}[1]
\Require Two sample sets: $\{ x_i \}_{i=1}^n \sim p(x)$, $\{ y_j \}_{j=1}^m \sim q(y)$; feature dimension $D$; kernel function parameters (e.g., bandwidth of Gaussian kernel $\sigma$)
\Ensure approximate MMD value $\text{MMD}_{\text{RFF}}$
\State \textbf{Step 1: Generate random frequencies and offsets}
\State Sample $D$ random frequency vectors $\{ \omega_d \}_{d=1}^D$ independently from the spectral density $p(\omega)$ corresponding to the kernel function. For Gaussian kernel, $\omega_d\sim \mathcal{N}(0,\frac{1}{\sigma^2}\mathbf{I})$.
\State Sample $D$ random offsets $\{ b_d \}_{d=1}^D$ independently from a uniform distribution $\mathcal{U}(0, 2\pi)$
\State \textbf{Step 2: Compute random Fourier feature mapping}
\For{$i = 1$ to $n$}
    \State For the $i$th sample $x_i$, compute its random Fourier feature vector:
    \[
    z(x_i) = \sqrt{\frac{2}{D}} \left[ \cos(\omega_1^\top x_i + b_1), \cos(\omega_2^\top x_i + b_2), \dotsc, \cos(\omega_D^\top x_i + b_D) \right]
    \]
\EndFor
\For{$j = 1$ to $m$}
    \State For the $j$th sample $y_j$, compute its random Fourier feature vector:
    \[
    z(y_j) = \sqrt{\frac{2}{D}} \left[ \cos(\omega_1^\top y_j + b_1), \cos(\omega_2^\top y_j + b_2), \dotsc, \cos(\omega_D^\top y_j + b_D) \right]
    \]
\EndFor
\State \textbf{Step 3: Calculate the mean embedding}
\State Calculate the mean embedding of the two distributions:
\[
\hat{\mu}_p = \frac{1}{n} \sum_{i=1}^n z(x_i), \quad \hat{\mu}_q = \frac{1}{m} \sum_{j=1}^m z(y_j)
\]
\State \textbf{Step 4: Calculate MMD value}
\State Calculate the approximate MMD value:
\[
\text{MMD}_{\text{RFF}} = \left\| \hat{\mu}_p - \hat{\mu}_q \right\|_2
\]
\State \Return $\text{MMD}_{\text{RFF}}$
\end{algorithmic}
\end{algorithm}

\section{Distance Measures}

In general, a statistical distance should satisfy several properties as follows (\cite{munkres2000topology}):

\begin{definition}[Distance Function]
    Let $X$ be a set. A function $d:X\times X \longrightarrow \mathbb{R}$ is a metric, if for all $x, y, z \in X$, the function $d$ follows the following conditions:
\end{definition}
\begin{itemize}
    \item Non negativity: for all $x,y \in X$, the metric $d$ is non-negative:\[
    d(x,y)\geq 0
    \]
    \item Positive definiteness:
    \[
    d(x,y) = 0 \iff x=y
    \]
    \item Symmetry: for all $x,y\in X$, the metric is symmetric:
    \[
    d(x,y) = d(y,x)
    \]
    \item Subadditivity: for all $x,y,z\in X$, the metric $d$ should fulfill the triangle inequality, which states that
    \[
    d(x,z) \leq d(x,y)+d(y,z).
    \]
    
\end{itemize}

\subsection{Technical Details and Illustration of the Wasserstein Distance}
\label{subsec:WS}

We show here some basic properties of the sliced Wasserstein distance (SWD). The majority of our proofs follow the seminal work established in \cite{bonnotte2013unidimensional}, which provides a comprehensive theoretical foundation for optimal transportation methods. The following theoretical framework establishes the fundamental characteristics of SWD, with particular emphasis on its metric properties and convergence behavior in probability spaces.\\

\begin{theorem}[Metric Property]

The Sliced Wasserstein Distance $SWD_p$ defines a proper metric on the space of probability measures.

\end{theorem}

\begin{proof} To prove that $SWD_p$ is a metric, we need to verify four properties: non-negativity, symmetry, triangular inequality, and positive definiteness.

Non-negativity and symmetry: Since $W_p$ is non-negative and symmetric, and the integral maintains non-negativity and symmetry, $SWD_p$ naturally satisfies these two properties.

Triangle inequality: Using the triangle inequality for $W_p$ and the linearity of the integral, it can be shown that $SWD_p$ satisfies the triangle inequality.

Positive definiteness: the key is to show that $SWD_p(\mu, \nu) = 0$ implies $\mu = \nu$. Assume $SWD_p(\mu, \nu) = 0$, then for every $\theta\in \mathbb{S}^{d-1}$, we have 
\[
W_p (\theta \mu,\theta \nu) = 0.
\]
This means that for all $\theta$, the projective distributions of $\mu$ and $\nu$ in the direction $\theta$ are the same.

Using the Fourier slice theorem, the Fourier transforms of $\mu$ and $\nu$ are the same in all directions:

\begin{equation}
    \mathcal{F}\mu(s\theta) = \mathcal{F}\left( \theta\mu \right)(s) = \mathcal{F}\left( \theta\nu \right)(s) = \mathcal{F}\nu(s\theta).
\end{equation}

Due to the invertibility of the Fourier transform on $\mathbb{R}^d$, this shows that $\mu = \nu$.
\end{proof}

\begin{lemma}[Equivalence Bounds \cite{bonnotte2013unidimensional}]
For $\mu,\nu \in \mathcal{P}p(\mathbb{R}^d)$, there exists a constant $c_{d,p}$ such that:

\begin{equation}
SWD_p(\mu,\nu)^p \leq c_{d,p}W_p(\mu,\nu)^p,
\end{equation}

where $c_{d,p} = \frac{1}{d}\int_{S^{d-1}} |\theta|_p^p d\theta \leq 1$.
\end{lemma}

Further, regarding the convergence rate, \cite{bonnotte2013unidimensional} shown the following Lemma. 

\begin{lemma}[Convergence Rate]
For measures supported in $B(0,R)$, there exists a constant $C_d > 0$ such that:
\begin{equation}
W_1(\mu,\nu) \leq C_dR^{d/(d+1)} SW_1(\mu,\nu)^{1/(d+1)}.
\end{equation}
\end{lemma}

For detailed proofs of Lemma 1 and Lemma 2, see \cite{bonnotte2013unidimensional}.\\

Next, we give some illustrations of how the projection process of the sliced Wasserstein distance works. 

\begin{figure}[h]
    \centering
    \includegraphics[width=\textwidth]{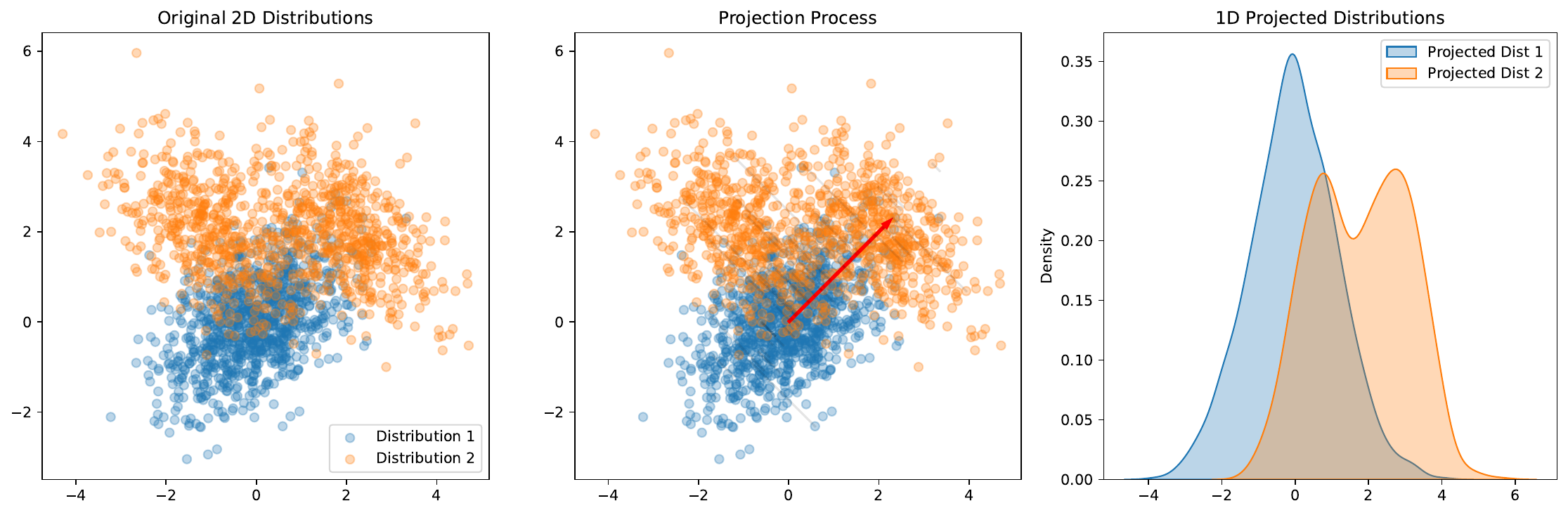}
    \caption{Illustration of the sliced Wasserstein distance computation process with 2D Gaussian distributions.}
    \label{fig:sw_illustration}
\end{figure}

 The left panel shows two original distributions in 2D space: a single Gaussian distribution (blue points) and a mixture of two Gaussians (orange points). The key insight of SWD is to project these high-dimensional distributions onto one-dimensional spaces. In the specific illustration here, we show a two-dimensional case, as shown by the red projection direction. The middle panel demonstrates this projection process, where each point is mapped onto the chosen direction. The right panel reveals the resulting one-dimensional probability densities after projection, effectively transforming our original distributional comparison problem into a simpler one-dimensional optimal transport problem.

\begin{figure}[h]
    \centering
    \includegraphics[width=\textwidth]{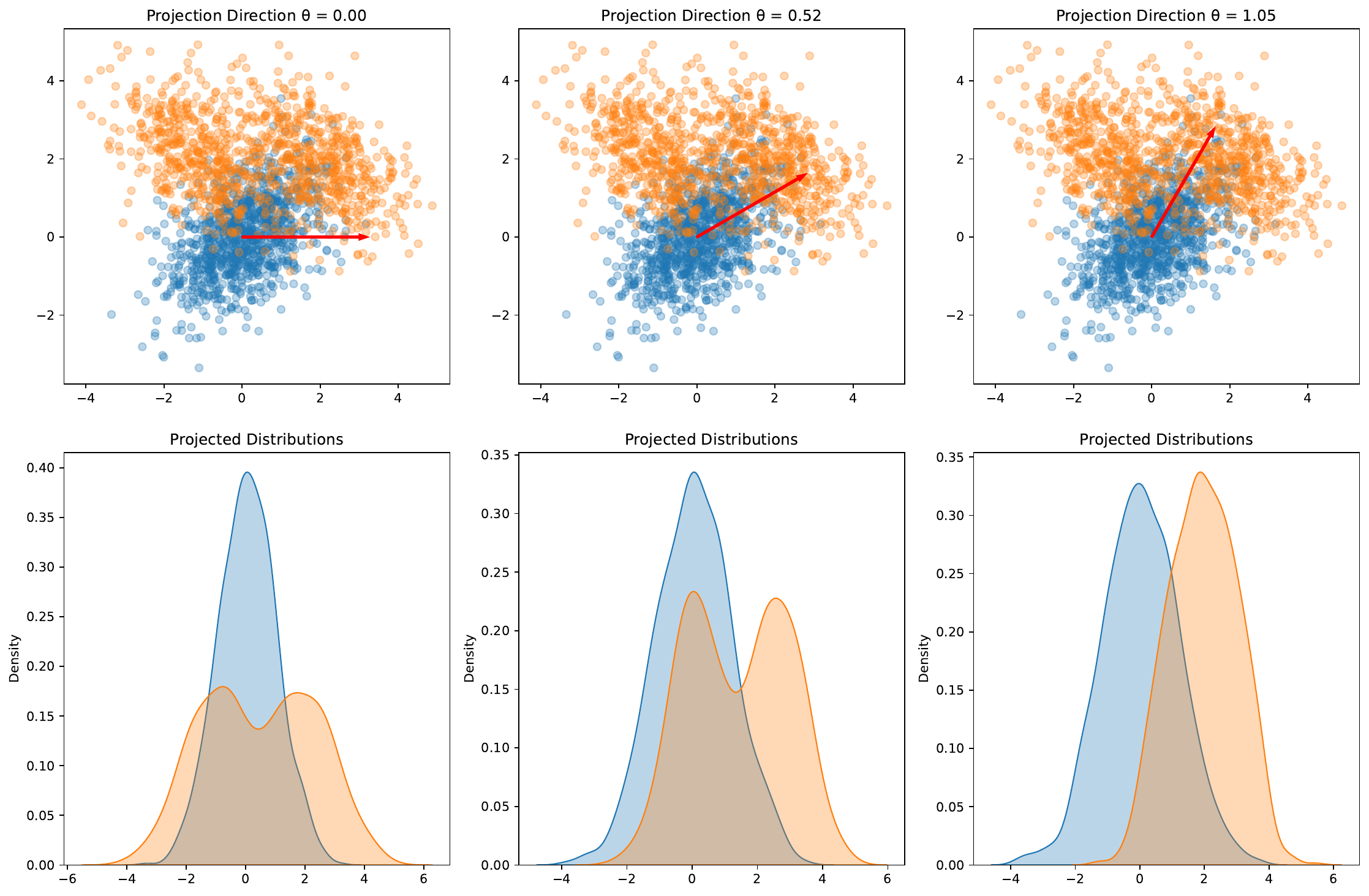}
    \caption{Effect of different projection angles ($\theta = 0, \pi/6, \pi/3$) on sliced Wasserstein distance computation. These directions are for visual illustrations only.}
    \label{fig:multiple_projections}
\end{figure}

The effectiveness of the SWD relies on using multiple random projections to capture different aspects of the distributional differences. As illustrated in Figure~\ref{fig:multiple_projections}, different projection angles reveal distinct characteristics of the distributions. When $\theta = 0$, the projection highlights the horizontal spread of the distributions, while larger angles progressively capture the vertical structural differences. This demonstrates why using multiple random projections is crucial for accurately computing the SWD, as each projection angle may capture different aspects of the distributional discrepancy.

\begin{figure}[t]
    \centering
    \begin{minipage}{0.48\textwidth}
        \centering
        \includegraphics[width=\textwidth]{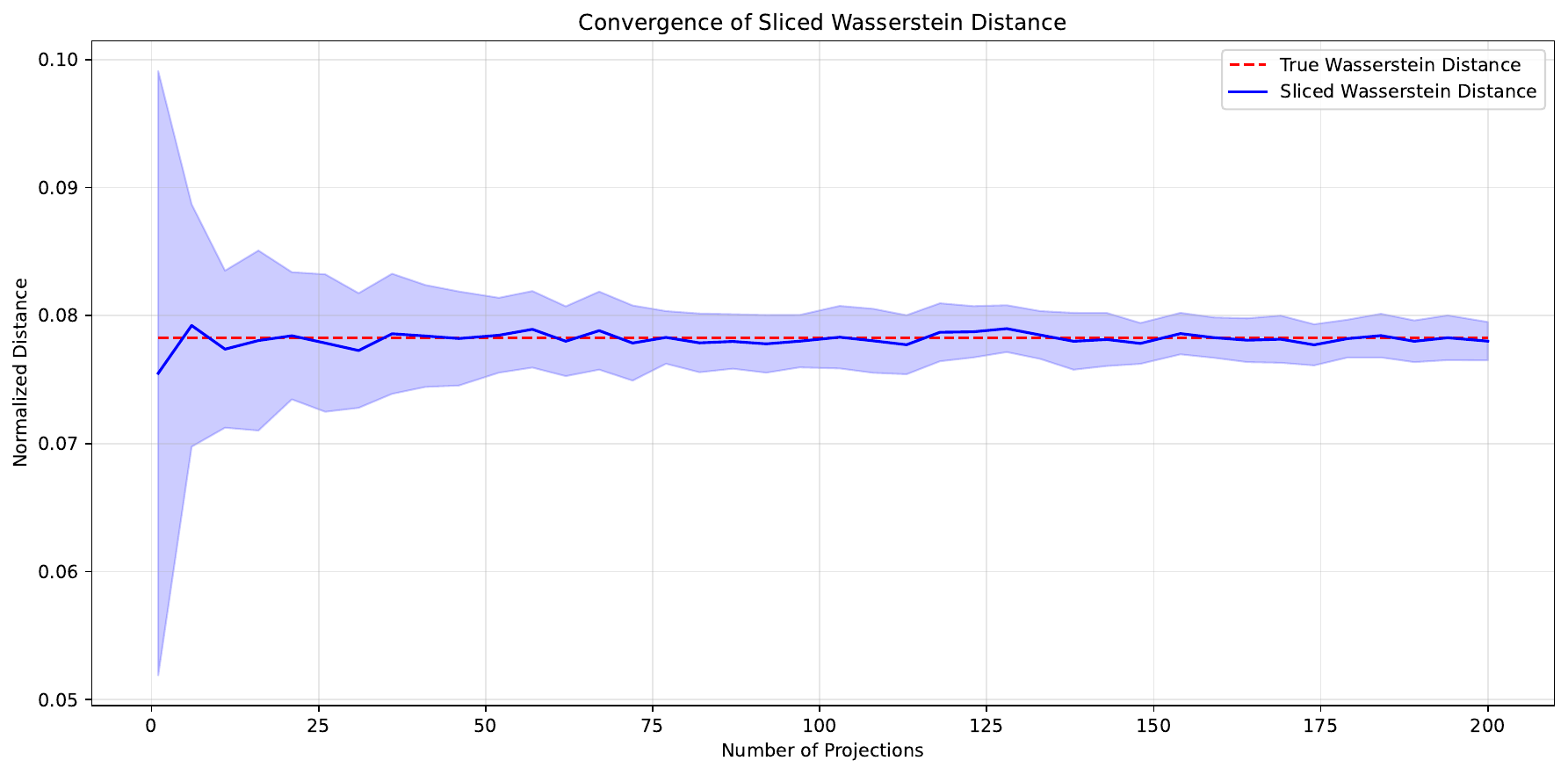}
        \caption*{(a) Convergence behavior of SWD}
    \end{minipage}
    \hfill 
    \begin{minipage}{0.48\textwidth}
        \centering
        \includegraphics[width=\textwidth]{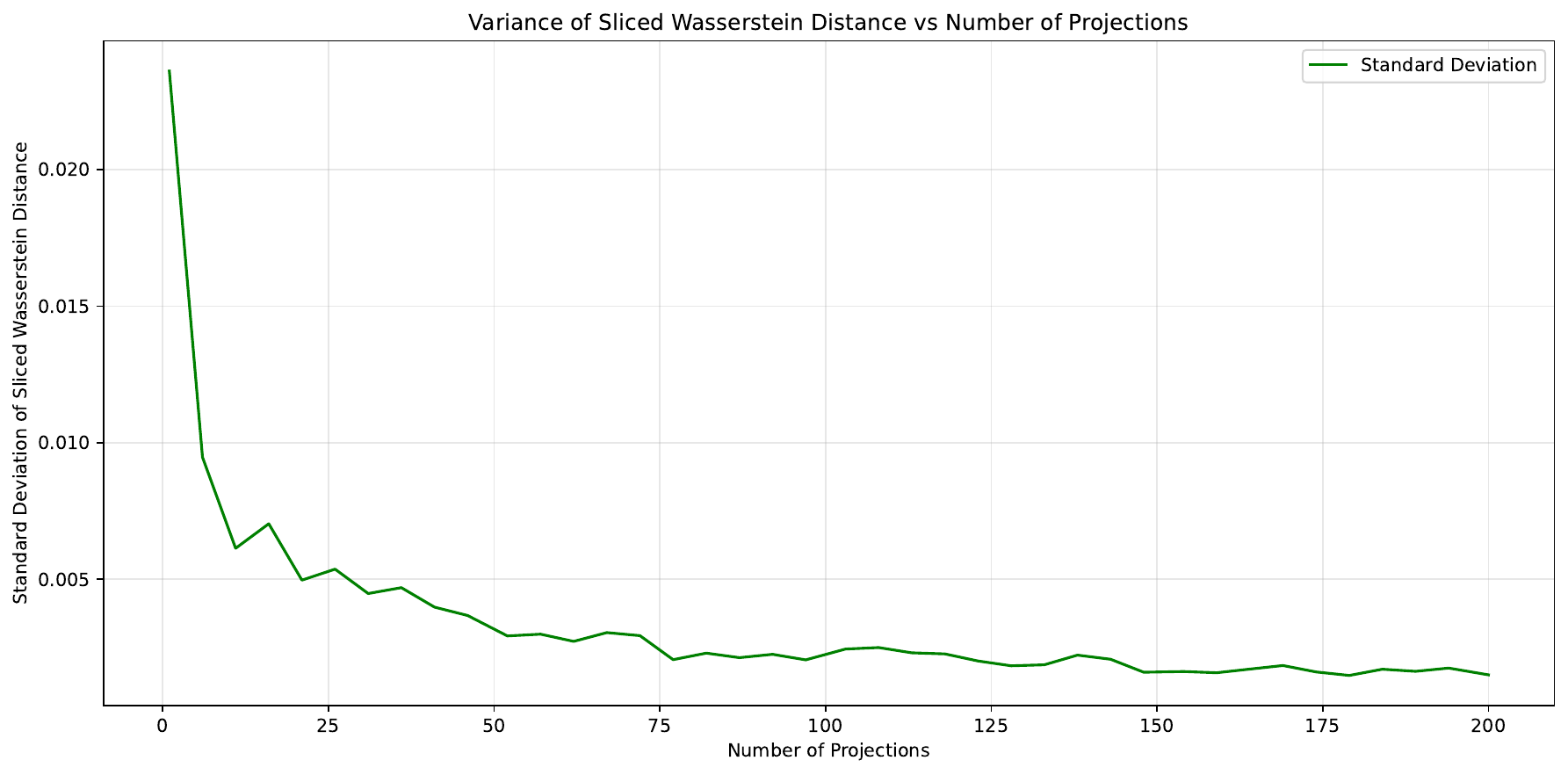}
        \caption*{(b) Variance reduction analysis}
    \end{minipage}
    
    \caption{Empirical analysis of sliced Wasserstein distance: (a) shows the convergence to true Wasserstein distance as the number of projections increases, with the blue shaded area representing one standard deviation; (b) demonstrates the decreasing trend of standard deviation with more projections, indicating improved stability of the estimation.}
    \label{fig:sw_convergence}
\end{figure}

Figure~\ref{fig:sw_convergence} provides an empirical analysis of the SWD. The convergence plot (a) demonstrates that the SWD rapidly approaches the true Wasserstein distance within the first 25 projections and stabilizes thereafter. The variance analysis (b) confirms this observation by showing a significant reduction in standard deviation as the number of projections increases. These results suggest that approximately 50 projections are sufficient for obtaining a reliable approximation while maintaining the computational efficiency of $O(Ln \log n)$, a substantial improvement over the $O(n^3 \log n)$ complexity of the exact Wasserstein computation (\cite{xie2020fast, le2024fast}). 

\subsection{Technical Details and Illustration of the Maximum Mean Discrepancy}
\label{subsec:MMD}

The maximum mean discrepancy (MMD) was first introduced by \cite{gretton2006kernel}, in which the authors introduced MMD as a measure to compare two distributions. It was then introduced in \cite{gretton2012kernel}, which includes detailed theoretical introduction and its statistical properties. 

In reproducing kernel Hilbert space (RKHS), the computation of MMD can actually be performed directly via the kernel function without explicitly computing the embedding in the high-dimensional feature space. The essence of the kernel trick is to implicitly compute in high-dimensional feature spaces by defining the kernel function $k(x,y)$ without having to explicitly compute the feature map $\phi(x)$. This allows us to compute efficiently in high-dimensional or even infinite-dimensional feature spaces.

It can be shown that the computational complexity of the kernel MMD is $O(n^2d)$ for the Gaussian kernel, where $n$ is the sample size and $d$ is the dimension of the data for two same sample size $n$ data. Since it is needed to calculate the pairwise distance of the kernel function $n^2$ times, the total time complexity is $O(n^2)\times O(d) = O(n^2d)$. In our empirical analysis, the memory requirements for MMD computation pose significant constraints, particularly when processing high-dimensional samples. When dealing with samples of dimension $d=100$, the computation becomes computationally intractable as $n$ exceeds $10,000$ samples, primarily due to the quadratic memory complexity O($n^2$) required for storing the kernel matrix.

Therefore, some approximation methods that can be used for fast computation and with lower complexity have been proposed.

For the Laplacian kernel, \cite{bodenham2023eummd} presents an efficient computational method, \texttt{euMMD}, for the computation of MMD statistics on one-dimensional data, reducing the complexity from $O(n^2)$ to $O(n \log n)$. In the high-dimensional case, the authors further borrowed the idea of the slicing idea of the SWD to map the high-dimensional data to one-dimensional by random projection, which is utilized for the approximation computation so as to maintain the computational efficiency.\\

\textbf{Random Fourier Features for MMD approximation}

For the Gaussian kernel, a popular approach is to use the Random Fourier Features (RFF) to approximate the kernel calculation. The theoretical justification for RFF comes from Bochner's theorem (\cite{bochner1959lectures}), which can be summarized as follows:\\

\begin{theorem}{Bochner's Theorem}

A continuous, shift-invariant kernel function $ k(\delta) = k(x - y) $ is positive definite if and only if it is the Fourier transform of some non-negative finite Borel measure $ \mu $, i.e.:
\[
k(\delta) = \int_{\mathbb{R}^d} e^{i \omega^\top \delta} \, d\mu(\omega),
\]
where $ \delta = x - y $, $ \omega \in \mathbb{R}^d $, and $ i $ are imaginary units.
\end{theorem}

Bochner's theorem shows that any continuous, shift-invariant, positive definite kernel function can be expressed as the Fourier transform of a non-negative finite Borel measure. Based on Bochner's theorem, we can represent the shift-invariant kernel function as a Fourier transform form. For example, for the Gaussian kernel function we have

\[
k(x, y) = \exp\left( -\frac{\| x - y \|^2}{2\sigma^2} \right).
\]

It can be expressed as:

\[
k(\delta) = \int_{\mathbb{R}^d} e^{i \omega^\top \delta} p(\omega) \, d\omega,
\]

where $p(\omega)$ is the Fourier transform of a Gaussian kernel $p(\omega) = \mathcal{N}(\omega; 0, 2\sigma^{-2} \mathbf{I})$ and let
\[
p(\omega) = (2\pi\sigma^{-2})^{-d/2}\exp(-\|\omega\|^2(\sigma^2/2)).
\]

RFF approximates the frequency $\{\omega_i\}$ by randomly sampling the frequency from the distribution $p(\omega)$ for $\omega$, thus approximating the otherwise infinite-dimensional integral representation to a finite-dimensional inner product representation.

In order to convert the complex exponential representation to the cosine representation, we sample a shift $\{b_i\}$ from the uniform distribution $\mathcal{U}(0,2\pi)$, which can be verified by:
\[
\mathbb{E}_b[\cos(\omega^Tx+b)\cos(\omega^Ty+b)] = \frac{1}{2}\cos(\omega^T(x-y)).
\]

By taking the real part and using Euler's formula $ e^{i \theta} = \cos \theta + i \sin \theta$, we have:

\[
k(x, y) = \int_{\mathbb{R}^d} 2 \cos\left( \omega^\top x + b \right) \cos\left( \omega^\top y + b \right) p(\omega) \, d\omega.
\]

Using the Monte Carlo integral, we have 
\[
k(x,y) \approx \frac{1}{D}\sum_{i=1}^D 2\cos(\omega^T+b_i).
\]

Finally, define the mapping $z:\mathbb{R}^d \rightarrow \mathbb{R}^D:$
\[
    z(x) = \sqrt{\frac{2}{D}} \left[ \cos(\omega_1^\top x + b_1), \cos(\omega_2^\top x + b_2), \dots, \cos(\omega_D^\top x + b_D) \right]^\top.
\]

With the above mapping, the kernel function can be approximated as:
\[
k(x, y) = \langle z(x), z(y) \rangle,
\]
where $ \langle \cdot, \cdot \rangle $ denotes the Euclidean inner product. In this way, the computation of the Gaussian kernel is transformed from the explicit $\exp(\frac{\|x-y\|^2}{2\sigma^2})$ to the computation of the inner product of two $D$-dimensional feature vectors.

To demonstrate the effectiveness and computational efficiency of MMD and RFF MMD in distributional difference measures, we designed a small illustration comparison experiment. 

\begin{figure}[h]
    \centering
    \includegraphics[width=\textwidth]{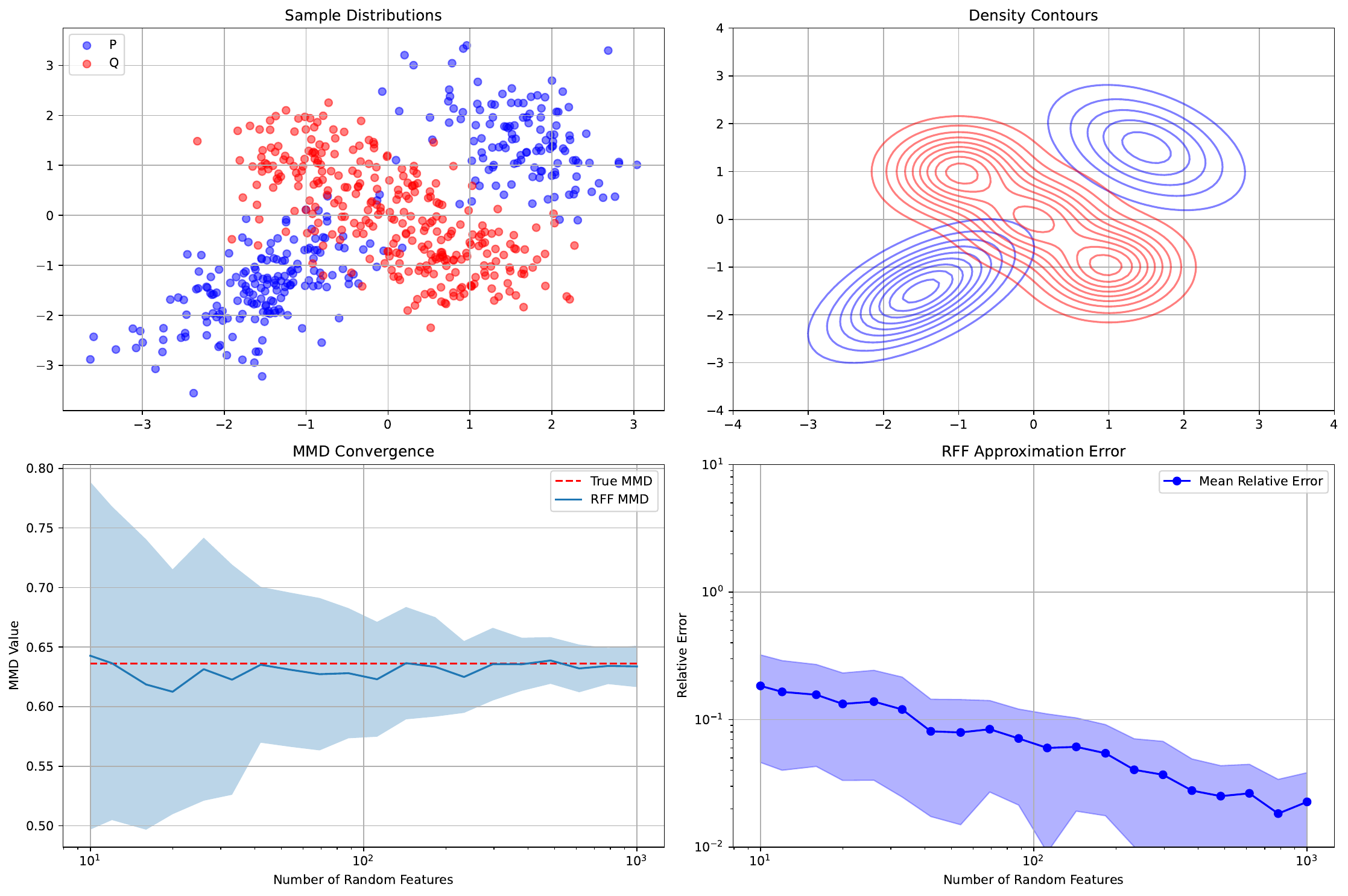}
    \caption{Illustration of the MMD and RFF-MMD comparison using two multi-modal Gaussian distributions.}
    \label{fig:MMD_illustration}
\end{figure}

As shown in the Figure \ref{fig:MMD_illustration}, the experiments use different multi-modal distributions $P$ and $Q$ that have clear spatial separation characteristics in the two-dimensional plane. The results show that the RFF MMD estimates exhibit good convergence properties as the random feature dimension $D$ is increased from 10 to 1000, with the relative error decreasing monotonically from 0.2 to 0.03. In particular, the confidence interval continues to narrow after $D > 100$, although the rate of error reduction tends to slow down, confirming that the RFF method is able to provide a reliable approximation of the MMD while maintaining computational efficiency.
\end{appendices}

\end{document}